\documentclass[10pt,journal,compsoc]{IEEEtran}
\usepackage{amssymb}
\usepackage{amsmath,amssymb,amsfonts}
\usepackage{multirow}
\usepackage{graphicx}
\usepackage{textcomp}
\usepackage{amsthm}
\renewcommand{\vec}[1]{\boldsymbol{#1}}
\usepackage{setspace}

\usepackage{graphicx}
\usepackage{algorithm}
\usepackage{algorithmic}
\usepackage{amsmath}
\usepackage{color}
\usepackage{amsfonts}

\ifCLASSOPTIONcompsoc
  \usepackage[nocompress]{cite}
\else
  \usepackage{cite}
\fi
\ifCLASSINFOpdf
\else
\fi

\newtheorem{thm}{Theorem}
\newtheorem{lem}{Lemma}

\newtheorem{defn}{Definition}

\newtheorem{exmp}{Example}
\newtheorem{rem}{Remark}
\newtheorem{pro}{Problem}

\begin{document}
%
\title{A k-hop Collaborate Game Model: Extended to Community Budgets and Adaptive Non-Submodularity}
%
%
%
%

\author{Jianxiong Guo,
	Weili Wu~\IEEEmembership{Member,~IEEE}
	\IEEEcompsocitemizethanks{\IEEEcompsocthanksitem J. Guo, and W. Wu are with the Department of Computer Science, Erik Jonsson School of Engineering and Computer Science, Univerity of Texas at Dallas, Richardson, TX, 75080 USA\protect\\
		E-mail: jianxiong.guo@utdallas.edu
	}
	\thanks{Manuscript received April 19, 2005; revised August 26, 2015.}}

%
%

\markboth{Journal of \LaTeX\ Class Files,~Vol.~14, No.~8, August~2015}%
{Shell \MakeLowercase{\textit{et al.}}: Bare Demo of IEEEtran.cls for Computer Society Journals}
%



\IEEEtitleabstractindextext{%
\begin{abstract}
Revenue maximization (RM) is one of the most important problems on online social networks (OSNs), which attempts to find a small subset of users in OSNs that makes the expected revenue maximized. It has been researched intensively before. However, most of exsiting literatures were based on non-adaptive seeding strategy and on simple information diffusion model, such as IC/LT-model. It considered the single influenced user as a measurement unit to quantify the revenue. Until Collaborate Game model \cite{guo2019khop} appeared, it considered activity as a basic object to compute the revenue. An activity initiated by a user can only influence those users whose distance are within k-hop from the initiator. Based on that, we adopt adaptive seed strategy and formulate the Revenue Maximization under the Size Budget (RMSB) problem. If taking into account the product's promotion, we extend RMSB to the Revenue Maximization under the Community Budget (RMCB) problem, where the influence can be distributed over the whole network. The objective function of RMSB and RMCB is adatpive monotone and not adaptive submodular, but in some special cases, it is adaptive submodular. We study the RMSB and RMCB problem under both the speical submodular cases and general non-submodular cases, and propose RMSBSolver and RMCBSolver to solve them with strong theoretical guarantees, respectively. Especially, we give a data-dependent approximation ratio for RMSB problem under the general non-submodular cases. Finally, we evaluate our proposed algorithms by conducting experiments on real datasets, and show the effectiveness and accuracy of our solutions.
\end{abstract}

\begin{IEEEkeywords}
Collaborate Game model, Online Social Networks (OSNs), Revenue maximization, Adaptive strategy, Non-submodular, Approximation Algorithm
\end{IEEEkeywords}}

\maketitle

\IEEEdisplaynontitleabstractindextext

%
\IEEEpeerreviewmaketitle

\IEEEraisesectionheading{\section{Introduction}\label{sec:introduction}}

%
%
%
%
\IEEEPARstart{T}{he} prosperous development of online social networks (OSNs) has derived a number of famous social platforms, such as Facebook, LinkedIn, Twitter and WeChat, which have become the main means of communication. There are more than 1.52 billion users active daily on Facebook and 321 million users active monthly on Twitter. These social platforms have been taken as an effective advertising channel by many companies to promote their products through "word-of-mouth" effect. It motivates the researches about viral marketing. The OSNs can be represented as an undirected graph, where nodes are users and edges denote the friendships between two users. Viral marketing, proposed by Domingos and Richardson \cite{domingos2001mining} \cite{richardson2002mining}, aims to maximize the follow-up users by giving rewards, coupons, or discounts to a subset of users who are the most influential. Then, Kempe et al. \cite{kempe2003maximizing} formulated Influence Maximization (IM) problem as a combinatorial problem: selects a subset of users as the seed set under the cardinality budget, such that the expected number of users who are influenced by this seed set can be maximized. They proposed two information diffusion models that were accepted by most researchers in subsequent researches: Independent Cascade model (IC-model) and Linear Threshold model (LT-model), and proved IM is NP-hard and its objective function is monotone submodular under these two models, thus, a good approximation can be obtained by natural greedy algorithm \cite{nemhauser1978analysis}. After this milestone work, a series of variant problems based on IM model and adapted to different application scenarios were emerged. For example, Revenue Maximization (RM), also called Profit Maximization, is a representative among them. RM problem is commonly studied in the literature \cite{arthur2009pricing} \cite{lu2012profit} \cite{tang2016profit} \cite{zhang2016profit} \cite{guo2019multi} \cite{8952599}. More different factors, such as product price, discount, cost \cite{zhou2015bilevel} \cite{lu2016pricing} \cite{yang2016continuous} \cite{ajorlou2016dynamic} and their impact on propagation, need to be considered in maximizing revenue.

From the perspective of companies, their objecitves are to maximize the exptected total revenue by selecting a subset of users from the whole network. However, most existing researches about RM problem were based on non-adaptive seeding strategy, where we select all seed nodes in one batch but do not consider the influence diffusion process. Thus, non-adaptive seeding is not the best choice to solve RM problem. Compared to that, adaptive one can response with a better seeding strategy because of making decision according to the real-time feedback from the users. Not only can it take advantage of the limited budget more wisely, but also adapt to the dynamic features of a social network. For example, the network topology is changing at any time because users join or leave, and friendships construct or terminate. In addition, most previous researches were based on the simple information diffusion model, such as IC-model or LT-model, where they considered the number of follow-up users that accept our information cascade. It applies user as a measurement unit to quantify the revenue, where each user corresponds to a fixed value of revenue. Sometimes, this is not enough and we need to apply activity as a measurement unit to quantify the revenue. Suppose some user in facebook is invited by a company to initiate an activity, after accepting and initiating this activity, the nerghbor users of this initiator would be infected. The initiator's friends, or friends of friends, may be participate in this activity. Thus, Guo et al. \cite{guo2019khop} propose Collaborate Game model to characterize this scene, where the total revenue gained by the company is correlated to the number of successfully initiated activity. The number of participants is different among different activities. For a single activity, the revenue we can obtain from this activity is related to the number of its participants, but not simple linear relationship. In general, the closer the distance is from the initiator, the more business benefit this participant will provide.

Based on the Collaborate Game model, we propose the Revenue Maximization under the Size Budget (RMSB) problem. RMSB aims to maximize the total revenue by inviting a small number of users to initiate an activity in an adaptive strategy. We adopt adaptive seeding strategy because: before sending an invitation to a potential initiator, we do not know whether she will accept it and how many users around her will follow and join together. By adaptive strategy, we can observe the actual state of users and edges after sending an invitation.
Besides, for a company, they should not only consider the largest revenue, but also consider the promotion effect of this product. According to that, we propose the Revenue Maximization under the Community Budget (RMCB) problem, where the number of invited users in each community of the targeted network is constrained. In this way, we can make sure there exists several activities happened in every area of the entire network, thus, the balance of influence distribution is guaranteed. Relied on the adaptive greedy policy, we design RMSBSolver and RMCBSolver algorithms to solve RMSB and RMCB respectively. The objective function of RMSB and RMCB problem is adaptive monotone and not adaptive submodular, but in some special cases, it is adaptive submodular. In this paper, the community budget of RMCB can be represented by use of partition matroid, and we can obtain a constant approximation ratio for RMCB problem under the special submodular cases. For the RMSB problem, it is an open question left by \cite{guo2019khop} how to get a theoretical bound for general non-submodular cases. In this paper, we attempt to solve it. Smith et al. \cite{smith2018approximately} proposed the concept of adaptive primal curvature and adaptive total primal curvature. We give a bound for adaptive total primal curvature and then, with the help of that, we can obtain a data-dependent approximation ratio for RMSB problem. Our contributions are summarized as follows:

\begin{enumerate}
	\item Based on Collaborate Game model, we propose RMSB and RMCB problem, which are adaptive monotone but not adaptive submodular.
	\item We design RMSBSolver and RMCBSolver algorithms to RMSB and RMCB problem, generalize RMCB to partition matroid, and obtain a $(1/2)$-approximation under the special submodular cases.
	\item To RMSB problem, we prove the solution returned by RMCBSolver satisfies a data-dependent $(1-e^{-\frac{1}{\delta}})$-approximation under the general non-submodular cases.
	\item Our proposed algorithms are evaluted on real-world datasets, and it shows that our algorithms are effective and better than other baseline algorithms.
\end{enumerate}

\textbf{Organiztion:} In Section 2, we survey the related work in RM and adaptive submodular optimization. We then present RMSB and RMCB problem in Section 3, discuss the algorthms in Section 4 and introduce the comprehensive theoretical analysis in Section 5. Finally, we conduct experiments and conclude in Section 6 and Section 7.

\section{Related Work}
Domingos and Richardson \cite{domingos2001mining} \cite{richardson2002mining} were the first to study viral marketing and the value of customers in social networks. Kempe et al. \cite{kempe2003maximizing} studied IM as a discrete optimization problem and generalized IC-model and LT-model to triggering model, who provided us with a greedy algorithm and a constant approximation ratio. RM is an important variant problem of IM, and its existing researches focused on in the non-adaptive setting. \cite{lu2012profit} \cite{tang2016profit} studied the problem that selects quality seed users such that maximizing revenue with the help of influence diffusion. Lu et al. \cite{lu2012profit} extended the LT-model to include prices and valuation, who used a heuristic unbedgeted greedy framework to solve this problem. Tang et al. \cite{tang2016profit} applied deterministic and randomized double greedy algorithms \cite{buchbinder2015tight} to solve the RM problem. If the objective function is non-negative and submodular, they obtained a $(1/3)$- and $(1/2)$-approximation ratio respectively. Zhang et al. \cite{zhang2016profit} investigated RM problem with multiple adoptions, which aimed at maximizing the overall profit across all products. Liu et al. \cite{liu2018profit} considered RM with coupons in his new model, independent cascade model with coupons and valuations (IC-CV), and solved it based on local search algorithm \cite{feige2011maximizing}. Recently, Tong et al. \cite{tong2018coupon} designed randomized algorithms, called simulation-based and realization-based RM, to address the RM with coupons and achieved a $(1/2)$-approximation with high probability. Guo et al. \cite{guo2019novel} proposed a composed influence model with complementary products and gave a solution by use of sandwich approximation framework \cite{lu2015competition}.

In the adaptive setting, Golovin et al. \cite{golovin2011adaptive} were the first to study the adaptive submodular optimization problem. Similar to the monotonicity and submodularity of set function, they extended these two concepts to adaptive version, adaptive monotonicity and adaptive submodularity, and proved that the solution returned by adaptive greedy policy is a $(1-1/e)$-approximation if the objective function is adaptive monotone and adaptive submodular. Applied it to social networks, Tong et al. \cite{tong2019adaptive} provided a systematic study on the adaptive influence maximization problem with different (partial or full) feedback model, especially for the algorithmic analysis of the scenarios when it is not adaptive submodular. Smith et al. \cite{smith2018approximately} introduced two important concepts, adaptive primal curvature and adaptive total primal curvature, to obtain a valid approximation ratio for adaptive greedy policy when the objective function is adaptive monotone but not adaptive submodular. Further, when the objective function is not adaptive monotone, but adaptive submodular, Gotoves et al. \cite{gotovos2015non} extended random greedy algorithm to adaptive policy and obtained a $(1/e)$-approximation. Other researches on the application of adaptive strategy can refer to \cite{gabillon2013adaptive} \cite{fern2017adaptive} \cite{yuan2017adaptive} \cite{han2018efficient}.

\section{Problem and Preliminaries}
In this section, we talk about some preliminary knowledges and concepts to the rest of paper, which maily includes model review and problem definition.

\subsection{Model Recapitulation}
The problem in this paper is an extension based on Collaborate Game model \cite{guo2019khop}, thus, we need to review it briefly here. If you want to know more details about Collaborate Game model, please read \cite{guo2019khop}. The Collaborate Game model is defined in the targeted network, which is an undirected graph $G=(V,E)$ where $V=\{v_1,v_2,...,v_n\}$ is the set of $n$ users and $E=\{e_1,e_2,...,e_m\}$ is the set of $m$ edges. If the edge $e=\{u,v\}$ exists, it means user $u$ and $v$ are friend, and there is a probability $p_e\in[0,1]$ representing their intimacy. For each user $v$, we denote $N(v)$ as the set of users who are the neighbor (friend) of $v$. 
\begin{defn}[Collaborate Game \cite{guo2019khop}]
When a game company invites a user $u\in V$ to play their game, she accepts this invitation to launch this game with probability $\theta_u$. If she accepts it, we call her as "initiator", and in round $1$, the initiator $u$ will invite each of her friend $v\in N(u)$ to participate in this game with success probability $p_{uv}$. We call user $u$ as 0-hop participant and user $v\in N(u)$ who accepts the invitation from $u$ as 1-hop participant. In round $i$, each (i-1)-hop participant $u'$ will invite each of her friend $v'\in N(u')$ to participate in this game with success probability $p_{u'v'}$. Then, the process terminates until finishing round $k$.
\end{defn}

Then, we have a "Tendency Assumption": assuming that a user $u$ is an i-hop participant to one initiator $x$ and a j-hop participant to another initiator $y$, where $i\leq j$, at this time, we consider user $u$ will choose to be an i-hop participant to initiator $x$. There are two additional parameters associated to our model: 
\begin{enumerate}
	\item Acceptance vector $\vec{\theta}=(\theta_1,\theta_2,...,\theta_n)$, where each element $\theta_u\in[0,1]$ is an acceptance probability for user $u\in V$. It quantifies the likelihood for user $u$ to accept to be an initiator when she receives the invitation from the company. In this paper, we assume $\theta_u$ is uniformly distributed in interval $[0,1]$.
	\item Revenue vector $\vec{R}=(R_0, R_1,R_2,...,R_k)$, where each element $R_i\in\mathbb{Z}^+$ is the revenue the company can gain from an i-hop participant. In this paper, we assume $R_0\geq R_1\geq R_2\geq...\geq R_k$ according to the "Benefit Diminishing" assumption \cite{guo2019khop}.
\end{enumerate}
\begin{rem}
The Collaborate Game model is based on such a scenario: A game company wants to promote their new multiplayer game over the targeted network by inviting some users to play this game. Once they accept it, they will attract their friends recursively to participate in this game. Then, game company can obtain revenue from it. Since we assume that the initiator's influence range is at most k-hop from her, this model is called k-hop Collaborate Game model as well. Thus, the "Benefit Diminishing" assumption means those participants whose are far from the initiator will provide less benefit to the company.
\end{rem}

From the perspective of companies, they are not sure whether she will accept it before sending an invitation to a potential initiator, and do not know how many users around her will follow and join together. Therefore, they need to adopt an adaptive strategy. Before determining who should be the next potential initiator, they need to make an observation to the change of states of both users and networks from the last invitation. Given targeted network $G=(V,E)$, for each user $u\in V$, the state of $u$ can be denoted by $X_u\in\{0,1\}\cup\{?\}$, where $X_u=1$ means user $u$ accepts to be an initiator because of the company's invitation and $X_u=0$ means user $u$ rejects to be an initiator. $X_u=?$ means user $u$ is unknown, who did not receive an invitation from the company. The states of all users are $?$ at the beginning. Similarly, for each edge $e=\{u,v\}\in E$, the state of $e$ can be denoted by $Y_e=\{0,1\}\cup\{?\}$, where $Y_e=0$ indicates user $u$ (resp. $v$) did not accept the invitation from user $v$ (resp. $u$) and $Y_e=1$ indicates user $u$ (resp. $v$) is willing to play the game with user $v$ (resp. $u$). Once determined, it cannot be changed. $Y_e=?$ indicates there is no invitation that happens between user $u$ and user $v$. The states of all edges are $?$ at the beginning.

\begin{figure}[!t]
	\centering
	\includegraphics[width=3.5in]{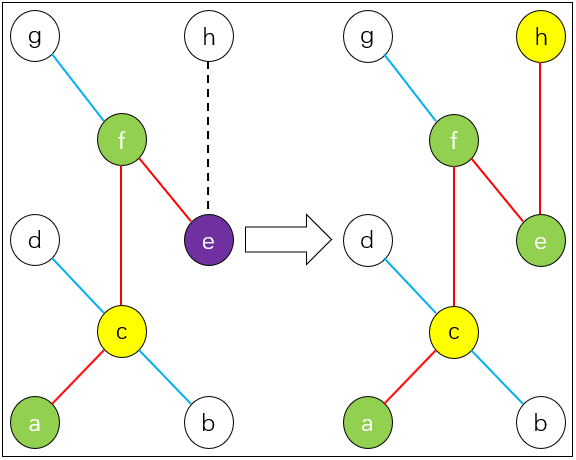}
	\caption{An example that shows the adaptive process: Here, we assume $k=2$, where the state of yellow nodes are $1$, other nodes are $?$; Green and purple nodes are 1-hop and 2-hop participants; The state of red, blue and dotted line are $1$, $0$ and $?$. First, we invite node $c$, the states shown as left part; Then, we invite node $h$, the states shown as right part. We can see that  node $e$ is changed from 2-hop to 1-hop participant because of node $h$. \cite{guo2019khop}}
	\label{fig1}
\end{figure}

After defining the states of user and edge, we have a function $\phi:\{X_u\}_{u\in V}\cup\{Y_e\}_{e\in E}\rightarrow U$ be all possible states, which is called as realization. Thus, we say that $\phi(u)$ is the state of user $u\in V$ and $\phi(e)$ is the state of edge $e\in E$ under the realization $\phi$. We denote by $\Phi$ a random realization and by $\Pr(\phi)=\Pr[\Phi=\phi]$ the probability distribution over all realizations. Besides, each realization should be consistent. Here, each user can only be one of states in $\in\{0,1\}\cup\{?\}$, and each edge can only be one of states in $\{0,1\}\cup\{?\}$ identically. After each pick, our observations so far can be represented as a partial realization $\psi$, which is a function from observed objects to their states. Then, $dom(\psi)$ is the domain of $\psi$, that is, the observed objects in $\psi$. A partial realization $\psi$ is consistent with a realization $\phi$ if they are equal everywhere in the domain of $\psi$, denoted by $\phi\sim\psi$. If $\psi$ and $\psi'$ are both consistent with some $\phi$, and $dom(\psi)\subseteq dom(\psi')$, we say $\psi$ is a subrealization of $\psi'$, denoted by $\psi\subseteq\psi'$. Besides, we denote the observed users in the domain of $\psi$ by $dx(\psi)$, and the observed edges in the domain of $\psi$ by $dy(\psi)$. \cite{guo2019khop}

\subsection{Problem Definition}
At each pick, the game company attempts to send an invitation to a user $u$ according to current partial realization $\psi$, and observe the $u$'s state $\Phi(u)$. If $\Phi(u)=1$, we need to update $\psi$, in other words, update the states of edges whose distance are less than k-hop from user $u$. In this process, user $u$ attracts her nearby users to play together. If $\Phi(u)=0$, we do nothing and go to the next pick. An example is shown as Fig. \ref{fig1}. Let $\pi$ be such an adaptive policy, and $H(\pi,\phi)$ be the set of users who are invited by the game company according to policy $\pi$ under the realization $\phi$. The total revenue gained according to policy $\pi$ under the realization $\phi$ can be defined as follows: \cite{guo2019khop}
\begin{equation}
f(H(\pi,\phi),\phi)=\sum_{i\in[k]}\sum_{u\in D_i(\pi,\phi)}R_i
\end{equation}
where $[k]=\{0,1,2,...,k\}$ and $D_i(\pi,\phi)$ is the set that contains all i-hop participants to the company according to policy $\pi$ under the realization $\phi$, thus, we have
\begin{flalign}
&D_0(\pi,\phi)=\{u|u\in H(\pi,\phi),\phi(u)=1\}\\
&D_i(\pi,\phi)=\nonumber\\
&\left\{u|\exists v\in D_{i-1}(\pi,\phi),\phi(\{u,v\})=1\right\}\backslash\bigcup_{j=0}^{i-1}D_j(\pi,\phi)
\end{flalign}
Finally, we can evaluate the performance of a policy $\pi$ by its expected revenue, and we have
\begin{equation}
f_{avg}(\pi)=\mathbb{E}_\Phi[f(H(\pi,\Phi),\Phi)]
\end{equation}
where the expectation is taken with respect to $\Pr(\Phi=\phi)$. Then, under the k-hop Collaborate Game model, the Revenue Maximization under the Size Budget (RMSB) \cite{guo2019khop} is formulated as follows:
\begin{pro}
Given a targeted network $G=(V,E)$, an acceptance vector $\vec{\theta}$, a size budget $b$ and a revenue vector $\vec{R}$, we aim to find a policy $\pi^*$ such that maximizing the expected revenue under the k-hop Collaborate Game model, $\pi^*\in\arg\max_{\pi}f_{avg}(\pi)$, subject to $|H(\pi,\phi)|\leq b$ for all realization $\phi$.
\end{pro}

However, a size budget is not enough because not only we want to maximize the revenue, but also need to consider the influence distribution. Generally, there exists a community structure given a social network, which is a partition of this network. Given a targeted network $G=(V,E)$, we assume it exists a unique community structure as $\mathcal{C}(G)$ associated with $G$, where $\mathcal{C}(G)=\{C_1,C_2,...,C_r\}$ is a partition of $V$. In other words, it means that $V=\bigcup_{i=1}^rC_i$ and $C_i\cap C_j=\emptyset$ for any $i,j\in\{1,2,...,r\}$. Denote community budget by vector $\vec{b}$, $\vec{b}=(b_1,b_2,...,b_r)$. It means that the number of users that invited to be initiators by the company in each community $C_i\in\mathcal{C}(G)$ cannot be larger than budget $b_i$. Then, under the k-hop Collaborate Game model, the Revenue Maximization under the Community Budget (RMCB) is formulated as follows:

\begin{pro}
	Given a targeted network $G=(V,E)$, a community structure $\mathcal{C}(G)$ associated with $G$, an acceptance vector $\vec{\theta}$, a community budget $\vec{b}$ and a revenue vector $\vec{R}$, we aim to find a policy $\pi^*$ such that maximizing the expected revenue under the k-hop Collaborate Game model, $\pi^*\in\arg\max_{\pi}f_{avg}(\pi)$, subject to $|H(\pi,\phi)\cap C_i|\leq b_i$ for any $C_i\in\mathcal{C}(G)$ and all realization $\phi$.
\end{pro}

Community budget in RMCB problem can be generalized to the matroid constraint, and we introduce some besic concepts about matroid here. A matroid $\mathcal{M}$ is an order pair $\mathcal{M}=(V,\mathcal{I})$, where $V$ is the ground set, node set in this paper, and $\mathcal{I}\subseteq 2^V$ is the collection of independent sets, which satisfies: (1) For all $A\subset B\subseteq V$, if $B\in\mathcal{I}$ then $A\in\mathcal{I}$; (2) For all $A,B\in\mathcal{I}$ with $|B|>|A|$, we have $\exists v\in B\backslash A$ such that $A\cup\{v\}\in\mathcal{I}$. Given a matroid $\mathcal{M}$, the matroid constraint means that a feasible solution $H(\pi,\phi)\in\mathcal{I}(\mathcal{M})$. The bases of matroid, denoted by $\mathcal{B}$, are those satisfy $\mathcal{B}\in\mathcal{I}$ and $\nexists v\in V$ such that $\mathcal{B}\cup\{v\}\in\mathcal{I}$. All bases of a matroid have the same size. There is a kind of special matroid, Partition Matroid, that is related to our problem: The ground set $V$ is partitioned into disjoint sets $C_1,C_2,...,C_r$, community structure in this paper, where $V=\bigcup_{i=1}^rC_i$. Let $b_i$, community budget in this paper, be an integer with $0\leq b_i\leq|C_i|$ for any $i\in\{1,2,...,r\}$. $I$ is an indepndent set, $I\in\mathcal{I}$, if and only if $|I\cap C_i|\leq b_i$ for any $i\in\{1,2,...,r\}$. The matroid whose independet sets are defined in this form is called partition matroid. Therefore, we want to find $\pi^*\in\arg\max_{\pi}f_{avg}(\pi)$ such that $H(\pi,\phi)\in\mathcal{I}$ for all realization $\phi$, and
\begin{equation}
\mathcal{I}=\{S\subseteq V:|S\cap C_i|\leq b_i \text{ for } i=1,2,...,r\}
\end{equation}
where $\mathcal{I}$ is a collection of independent sets defined by partition matroid we talk about before.

\section{Algorithm}
In this section, we propose our algorithms to solve RMSB and RMCB problem. Adaptive greedy policy, proposed by Golovin et al. \cite{golovin2011adaptive}, is an effective method to solve such problems. It can be divided into two steps at each iteration: 
\begin{enumerate}
	\item Send an invitation to the user who has the largest expected increment of revenue based on the partial realization $\psi$.
	\item Observe the state of this user, if she accepts this invitation, update the states of edges whose distance within k-hop from her; Otherwise, back to (1) to invite the next user.
\end{enumerate}
The Conditional Exepected Marginal Benefit under the adaptive surrounding can be shown as follows:
\begin{defn}[Conditional Expected Marginal Benefit]
	Given a partial realization $\psi$ and an user $u$, the conditional expected marginal benefit of $u$ conditioned on observed $\psi$ is
	\begin{equation}
	\Delta(u|\psi)=\mathbb{E}\left[f(dx(\psi)\cup\{u\},\Phi)-f(dx(\psi),\Phi)|\Phi\sim\psi\right]
	\end{equation}
\end{defn}
\noindent
In our RMSB and RMCB problem, $\Delta(u|\psi)$ is the expected revenue increment of user $u$ by sending an invitation to initiate user $u$. It conditions on the previous observation to users and edges, partial realization $\psi$, and the expectation is taken over all realization that are consistant with $\psi$.

Based on adaptive greedy policy \cite{golovin2011adaptive}, RMSBSolver and RMCBSolver is proposed, shown as Algorithm \ref{a1} and Algorihtm \ref{a2}. In algorithm \ref{a1}, we invite a user with the largest expected revenue increment at each iteration until the number of invitation is larger than size budget $b$. However, in Algorithm \ref{a2}, it is more complicated. At each iteration, we need to check whether the candidate set $I$ is empty. This candidate set is determined by the collection of independent set $\mathcal{I}$, shown as Equation (5), given community budget $\vec{b}$. If this candidate set is empty, it means that the number of invitations at each community reaches its budget constraint, thus, stop inviting at that time.

\begin{algorithm}[!t]
	\caption{\textbf{RMSBSolver $(G,f,\vec{\theta},\vec{R},b,k)$}}\label{a1}
	\begin{algorithmic}[1]
		\STATE Initialize: $H\leftarrow\emptyset$, $\psi\leftarrow\emptyset$
		\FOR {$i=1$ to $b$}
		\FOR {user $u\in V\backslash H$}
		\STATE $\Delta(u|\psi)\leftarrow$ Computed by Monte-Carlo simulation or simplified method of Section 4.2 in \cite{guo2019khop}
		\ENDFOR
		\STATE Select $u_i\in\arg\max_{u\in V\backslash H}\Delta(u|\psi)$
		\STATE $H\leftarrow H\cup\{u_i\}$
		\IF {$u_i$ accepts the invitation}
		\STATE Update $\psi$, the states of the edges whose distance are less than k-hop from $u_i$
		\ENDIF
		\ENDFOR
		\RETURN $H,\psi$
	\end{algorithmic}
\end{algorithm}

At each iteration, we need to compute the value of $\Delta(u|\psi)$ given current partial realization, which can be done by Monte-Carlo simulation: running this diffusion process many times and then take the average of them. In order to achieve satisfactory accuracy, the number of simulations cannot be too small, however, it increases the running time greatly. To improve its scabability, a simplified computing method was proposed in Algorithm \ref{a2}, Section 4.2 of \cite{guo2019khop}. Its core idea is computing the value of $\Delta(u|\psi)$ by use of a technique similar to breath-first searching instead of Monto carlo simulation. It fixes all potential i-hop participants for $i\in[k]$ in advance, ignores some cases where the probability of occurrence is low and computes the probability of potential participants to join this game from up to down. The time complexity to compute each $\Delta(u|\psi)$ is reduced to $O(mn)$. The effectiveness of this method was validated by experiments on real-world dataset, and it reduces the running time significantly under the premise of ensuring accuracy. Besides, similar to \cite{guo2019khop}, we do not need to compute $\Delta(u|\psi)$ for each user $u\in I$ repeatedly at each iteration. If, at iteration $i$, the selected user $u_{i-1}$ at the last iteration rejects the invitation, we do not need to update any $\Delta(u|\psi)$. Otherwise, we only need to update $\Delta(u|\psi)$ for those users whose distences from $u_{i-1}$ is less than or equal to $2k$ at iteration $i$. It improves the efficieny of Algorithm \ref{a1} and Algorithm \ref{a2} further.

\begin{algorithm}[!t]
	\caption{\textbf{RMCBSolver $(G,\mathcal{C}(G), f,\vec{\theta},\vec{R},\vec{b},k)$}}\label{a2}
	\begin{algorithmic}[1]
		\STATE Initialize: $H\leftarrow\emptyset$, $\psi\leftarrow\emptyset$
		\WHILE {$True$}
		\STATE $I\leftarrow\{v\in V\backslash H:H\cup\{v\}\in\mathcal{I}\}$
		\IF {$I=\emptyset$}
		\STATE break
		\ENDIF
		\FOR {user $u\in I$}
		\STATE $\Delta(u|\psi)\leftarrow$ Computed by Monte Carlo simulation or simplified method of Section 4.2 in \cite{guo2019khop}
		\ENDFOR
		\STATE Select $u_i\in\arg\max_{u\in I}\Delta(u|\psi)$
		\STATE $H\leftarrow H\cup\{u_i\}$
		\IF {$u_i$ accepts the invitation}
		\STATE Update $\psi$, the states of the edges whose distance are less than k-hop from $u_i$
		\ENDIF
		\ENDWHILE
		\RETURN $H,\psi$
	\end{algorithmic}
\end{algorithm}

\section{Theoretical Analysis}
In this section, we discuss some related conclusions about RMSB and RMCB problem, and then extend them to generalized non-submodular cases.
\subsection{Related Conclusions}
Before starting our discussion, we introduce two important concepts, defined in \cite{golovin2011adaptive}, as follows:
\begin{defn}[Adaptive Monotone]
	A function $f$ is adaptive monotone with respect to distribution $\Pr(\phi)$ if the conditional expected marginal benefit of any user $u$ is nonnegative, i.e., for all $\psi$ with $\Pr[\Phi\sim\psi]>0$ and all $u\in V(G)\backslash dx(\psi)$, we have
	\begin{equation}
	\Delta(u|\psi)\geq 0
	\end{equation}
\end{defn}
\begin{defn}[Adaptive Submodular]
	A function $f$ is adaptive submodular with respect to distribution $\Pr(\phi)$ if the conditional expected marginal benefit of any user $u$ does not increase as more states of users and edges are observed, i.e., for all $\psi$ and $\psi'$ such that $\psi$ is a subrealization of $\psi'$ ($\psi\subseteq \psi'$) and all $u\in V(G)\backslash dx(\psi')$, we have
	\begin{equation}
	\Delta(u|\psi)\geq\Delta(u|\psi')
	\end{equation}
\end{defn}
Our RMSB and RMCB problem are two instances of adaptive optimization problem, thus, determined theoretical bounds can be obtained for both instances by adaptive greedy policy if our objective function, Equation (4), is adapative monotone and adaptive submodular. Based on \cite{guo2019khop}, we have seveal conclusions as follows:
\begin{thm}
	The RMSB and RMCB problem is NP-hard, and there are no solutions with polynomial time unless NP$=$P.
\end{thm}
\begin{proof}
	The RMSB is NP-hard \cite{guo2019khop}, because it can be reduced to Maximum Coverage problem when setting revenue vector $\vec{R}=\vec{1}$, acceptance vector $\vec{\theta}=\vec{1}$ and $p_e=1$ for each $e\in E$. Then, when there is only one community in targeted network $G$, and community budget is a scalar, at this moment, RMCB can be reduced to RMSB problem. Thus, RMCB problem is NP-hard as well.
\end{proof}
\begin{lem}
	The objective function $f$ of the RMSB and RMCB problem is adaptive monotone.
\end{lem}
\begin{lem}
	The objective function $f$ of the RMSB and RMCB problem is not adaptive submodular.
\end{lem}
\begin{lem}[\cite{guo2019khop}]
	If the RMSB and RMCB problem conform one of following two special cases:
	\begin{enumerate}
		\item $k\leq 1$
		\item For all $\{u,v\}\in E$, $p_{uv}=1$
	\end{enumerate}
	The objective function $f$ is adaptive submodular.
\end{lem}
\noindent
Those properties suitable to RMSB problem can be applied to RMCB problem, because their objective function is identical, but constraints are different. Therefore, for RMSB problem, we have
\begin{thm}
	The adaptive greedy policy $\pi_s$, given by Algorithm \ref{a1}, for RMSB problem is a $(1-1/e)$-approximate solution when $k\leq 1$ or $p_e=1$ for each $e\in E$. Hence, we have
	\begin{equation}
	f_{avg}(\pi^s)\geq(1-1/e)\cdot f_{avg}(\pi^*)
	\end{equation}
	where $\pi^*$ is optimal policy s.t. $|H(\pi^*,\phi)|\leq b$ for all $\phi$.
\end{thm}
\begin{proof}
	From Lemma 1, the objective function $f$ is adaptive monotone, and from Lemma 3, $f$ is adaptive submodular when $k\leq 1$ or $p_e=1$ for each $e\in E$. The adaptive greedy policy is a $(1-1/e)$-approximation under the size constraint according to the conclusion of \cite{golovin2011adaptive}.
\end{proof}

In \cite{golovin2011matroid}, Golovin et al. analyzed the theoretical performance of adaptive submodular optimization under p-independence system. The main conclusions are summarized as follows:
\begin{lem}[\cite{golovin2011matroid}]
Given an adaptive monotone and submodular function $f$ and a p-independence system $(V,\mathcal{I})$, the adaptive $\alpha$-approximate greedy policy $\pi$ under the constraint $\mathcal{I}$ yields and $\alpha/(p+\alpha)$-approximation. Hence, we have
\begin{equation}
f_{avg}(\pi)\geq\left(\frac{\alpha}{p+\alpha}\right)\cdot f_{avg}(\pi^*)
\end{equation}
where $\pi^*$ is optimal policy s.t. $H(\pi^*,\phi)\in\mathcal{I}$ for all $\phi$.
\end{lem}
\noindent
Therefore, for RMCB problem, we have
\begin{thm}
	The adaptive greedy policy $\pi_c$, given by Algorithm \ref{a2}, for RMCB problem is a $(1/2)$-approximate solution when $k\leq 1$ or $p_e=1$ for each $e\in E$. Hence, we have
	\begin{equation}
	f_{avg}(\pi^c)\geq(1/2)\cdot f_{avg}(\pi^*)
	\end{equation}
	where $\pi^*$ is optimal policy s.t. $|H(\pi^*,\phi)\cap C_i|\leq b_i$ for any $i\in\{1,2,...,r\}$ and all $\phi$.
\end{thm}
\begin{proof}
From Lemma 1, the objective function $f$ is adaptive monotone, and from Lemma 3, $f$ is adaptive submodular when $k\leq 1$ or $p_e=1$ for each $e\in E$. As we said before, the constraint of community budget can be classified to partition matroid. The independent set $\mathcal{I}$ in Lemma 4 can be defined by Equation (5). It is 1-independence system, thus, we have $p=1$. And the adaptive greedy policy is exact, which means that $\alpha=1$. According to Lemma 4, Equation (10), we have $\alpha/(p+\alpha)=1/2$ finally.
\end{proof}

\subsection{Non-Submodularity}
In addition to the special cases we discuss in the last subsection, the objective function of RMSB problem is not adaptive submodular, thus, the approximation ratio in Theorem 2 cannot be held in general cases. In order to deal with the general cases that is not adaptive submodular, we get help from the concept of adaptive primal curvature, proposed by \cite{smith2018approximately}, to obtain the approximation bound for adaptive greedy policy without adaptive submodularity. Wang et al. \cite{wang2016approximation} were the first to propose the concept of elemental curvature, which is the maximum ratio of the marginal gain of an element $i$ at any subset $S$ and $S\cup\{j\}$ for any element $j$. Extended to adaptive case, adaptive primal curvature \cite{smith2018approximately} was formulated, which is the ratio of the marginal gain of element $i$ under partial realization $\psi$ and $\psi\cup\{s\}$, where $s$ is possible state for any element $j$. Thus, we have
\begin{defn}[Adaptive Primal Curvature \cite{smith2018approximately}]
Given a adaptive momotone function $f$, the adaptive primal curvature is
\begin{equation}
\nabla_f(i,j|\psi)=\mathbb{E}\left[\frac{\Delta(i|\psi\cup s)}{\Delta(i|\psi)}\bigg|s\in S(j)\right]
\end{equation}
where $S(j)$ is the state set of $j$ and $\Delta$ is conditional expected marginal benefit, defined as Definition 2.
\end{defn}

From Equation (12), the $\nabla_f(i,j|\psi)$ measures the changing of the conditional expected marginal benefit because of $j$ added to the solution previously. Under the background of RMSB and RMCB problem, the meaning of $\Delta(i|\psi\cup s)$ is a little different. Here, $i$ and $j$ represent the users that are invited to be initiators. Let us look at the term: $\psi\cup s$, $s(j)$ can be $0$ or $1$. When $s(j)=0$, $\psi\cup s$ means we only add state $s(j)=0$ to $\psi$. But when $s(j)=1$, the meaning of $\psi\cup s$ will be more complicated, where we not only add the state $s(j)=1$ to $\psi$, but also all state change of edges since $j$ becomes an initiator should be added to $\psi$. Thus, the expectation of primal curvature is with respect to the acceptance probability of $i$ and $j$, and edge probabilities. In the non-adaptive case, $s(j)=1$ is determined, the elemental curvature is the maximum primal curvature. To measure the total change form partial realization $\psi$ to $\psi'$, total primal curvature \cite{smith2018approximately} is defined as follows:
\begin{defn}[Adaptive Total Primal Curvature \cite{smith2018approximately}]
Let $\psi\subset\psi'$ and let $\psi\rightarrow\psi'$ represents the set of possible state sequences changing from partial realization $\psi$ to $\psi'$. Then the adaptive total primal curvature is
\begin{equation*}
\Gamma(i|\psi',\psi)=\mathbb{E}\left[\prod_{s_j\in Q}\nabla'\left(i,s_j|\psi\cup[s_{j-1}]\right)\Bigg|Q\in\psi\rightarrow\psi'\right]
\end{equation*}
where $[s_{j-1}]=\{s_1,...,s_{j-1}\}$ and
\begin{equation*}
\nabla'(i,s_j|\psi)=\frac{\Delta(i|\psi\cup s_j)}{\Delta(i|\psi)}
\end{equation*}
which is a simplified notation of the adaptive primal curvature corresponded to a single state $s_j\in S(j)$.
\end{defn}
In the following, we can find a constant upper bound of the adaptive total primal curvature for the RMSB problem and a relationship between any policy with budget $b$ and adaptive greedy policy with any budget. Then, according to that, the approximation ratio of general RMSB problem can be found later.
\begin{lem}
The adaptive total primal curvature $\Gamma(i|\psi',\psi)$ for any $i$, $\psi$ and $\psi'$ is upper bounded by $\delta$, $\Gamma(i|\psi',\psi)\leq\delta$ in the general cases for RMSB problem, where
\begin{equation}
\delta=\frac{R_0+(n-1)\cdot R_1}{R_0-R_1}
\end{equation}
\end{lem}
\begin{proof}
The adaptive total primal curvature $\Gamma(i|\psi',\psi)$ can be denoted by $\Gamma(i|\psi',\psi)=\Delta(i|\psi')/\Delta(i|\psi)$. Let us assume sequence $\{s_1,s_2,...,s_l\}\in\phi\rightarrow\phi'$, we have
\begin{equation*}
\frac{\Delta(i|\psi\cup\{s_1\})}{\Delta(i|\psi)}\cdot\frac{\Delta(i|\psi\cup\{s_1,s_2\})}{\Delta(i|\psi\cup\{s_1\})}\cdot\cdot\cdot\frac{\Delta(i|\psi')}{\Delta(i|\psi\cup[s_{l-1}])}
\end{equation*}
From above, the product for any sequence from partial realization $\psi$ to $\psi'$ can be reduced to $\Delta(i|\psi')/\Delta(i|\psi)$ trivially. Then, we have
\begin{equation}
\Gamma(i|\psi',\psi)\leq\frac{\theta_i\cdot\max_{\psi}\Delta(i|\psi,s(i)=1)}{\theta_i\cdot\min_{\psi}\Delta(i|\psi,s(i)=1)}
\end{equation}
where $\Delta(i|\psi,s(i)=1)$ is the marginal gain assuming that user $i$ accepts to be the initiator. For $\max_{\psi}\Delta(i|\psi,s(i)=1)$, the revenue from user $i$ is at most $R_0$ and from each of user in $V\backslash\{i\}$ is at most $R_1$. Considering an extreme example, there exists an edge from user $i$ to each of user $j\in V\backslash\{i\}$ and the edge probability is equal to $1$ for each edge. Thus, the total revenue because of selecting $i$ as initiator could be at most $R_0+(n-1)\cdot R_1$. For $\min_{\psi}\Delta(i|\psi,s(i)=1)$, we need to know the current participation role of user $i$, maybe user $i$ is a x-hop participant, $x\in\{1,2,...,k\}$ before selecting $i$ as initiator. The revenue from user $i$ is equal to $R_0-R_x$. Based on Benefit diminishing assumption, $R_1\geq R_2\geq...\geq R_k$, the revenue from user $i$ is at least $R_0-R_1$. Considering an extreme example, all neighbors of user $i$ have been 1-hop participants before selecting $i$ as initiator, there is no revenue for other users. Thus, the total revenue because of selecting $i$ as initiator could be at most $R_0-R_1$. Combined with the above analysis, the Lemma is held.
\end{proof}

Let us denote by $\pi_b$ any policy that inviting exact $b$ users, and $\pi_l^g$ by the adaptive greedy policy that inviting $l$ users. Besides, we define $g_l$ as the $l^{th}$ user who is invited to be an initiator in adaptive greedy policy. Then, we can bound the adaptive greedy policy that inviting $l$ users with any policy that inviting exact $b$ users.
\begin{lem}
The difference of expected revenue bewtween any policy $\pi_b$ and adaptive greedy policy $\pi_l^g$ in the general cases for RMSB problem can be bounded by
\begin{equation}
f_{avg}(\pi_b)-f_{avg}(\pi_l^g)\leq b\delta\cdot\Delta_{l+1}
\end{equation}
where we assume $l<b$ and $\Delta_{l+1}$ is defined as $\Delta_{l+1}=f_{avg}(\pi_{l+1}^g)-f_{avg}(\pi_l^g)$.
\end{lem}
\begin{proof}
We know $f_{avg}(\pi_b)\leq f_{avg}(\pi_l^g@\pi_b)$ because $f_{avg}(\cdot)$ is adaptive monotone \cite{golovin2011adaptive}. From it, we have
\begin{equation*}
f_{avg}(\pi_b)-f_{avg}(\pi_l^g)\leq f_{avg}(\pi_l^g@\pi_b)-f_{avg}(\pi_l^g)
\end{equation*}
The difference of expected revenue between $\pi_b$ and $\pi_l^g$ can be bounded by running $\pi_b$ after running $\pi_l^g$. It means that we need to send $b$ invitations again. Supposing that $\psi$ is the partial realization generated by $\pi_l^g$, the initiators selected by $\pi_b$ is on partial realization $\psi'$ such that $\psi\subset\psi'$. Then, we have $\Delta(i|\psi')\leq\delta\Delta(g_{l+1}|\psi)$ if $\psi\subset\psi'$ and $i\notin dx(\psi)$, because $\Delta(i|\psi')=\Gamma(i|\psi',\psi)\cdot \Delta(i|\psi)\leq\delta\Delta(g_l|\psi)$ from Lemma 5. Therefore,

\begin{flalign}
&f_{avg}(\pi_b)-f_{avg}(\pi_l^g)\nonumber\\
&\leq b\delta\cdot\mathbb{E}[\Delta(g_{l+1}|\psi)|\psi]\nonumber\\
&=b\delta\cdot\mathbb{E}[\mathbb{E}[f(dx(\psi)\cup\{g_{l+1}\},\Phi)-f(dx(\psi),\Psi)|\Phi\sim\psi]|\psi]\nonumber\\
&=b\delta\cdot\mathbb{E}[f(H(\pi_{l+1}^g,\Phi),\Phi)-f(H(\pi_{l}^g,\Phi),\Phi)]\nonumber\\
&=b\delta\cdot f_{avg}(\pi_{l+1}^g)-f_{avg}(\pi_l^g)\nonumber\\
&=b\delta\cdot\Delta_{l+1}\nonumber
\end{flalign}
where the first inequality is from seeding $b$ invitations again, and other equality from Definition 2 and Equation (4).
\end{proof}
So far, we can obtain the main theorem in this paper that gets the approximation of RMSB problem in the general cases that is not adaptive submodular.
\begin{thm}
The approximation performance of Adaptive greedy policy $\pi_b^g$, given by Algorithm \ref{a1}, in the general cases for RMSB problem satisfies the following:
\begin{equation}
f_{avg}(\pi_b^g)\geq\left(1-e^{-\frac{1}{\delta}}\right)\cdot f_{avg}(\pi_b)
\end{equation}
when we hypothesize $\Delta_{b+1}\leq(1-1/b\delta)^b\cdot\Delta_1$.
\end{thm}
\begin{proof}
According to the Lemma 6, $f_{avg}(\pi_b)-f_{avg}(\pi_l^g)$ can be bounded by $b\delta\cdot\Delta_{l+1}$, we have
\begin{flalign}
f_{avg}(\pi_b)&\leq f_{avg}(\pi_l^g)+b\delta\cdot\Delta_{l+1}\\
&=\sum_{i=1}^{l}\Delta_l+b\delta\cdot\Delta_{l+1}
\end{flalign}
where $f_{avg}(\pi_l^g)=\sum_{i=1}^{l}\Delta_l$ and $\Delta_{l}=f_{avg}(\pi_{l}^g)-f_{avg}(\pi_{l-1}^g)$. We note that $f_{avg}(\pi_0^g)=0$. Multiply both side by $(1-(b\delta)^{-1})^{b-l}$ of Inequality (17) and sum from $l=1$ to $b$. The left hand side of Inequality (17) can be reduced to
\begin{equation*}
\sum_{l=1}^{b}\left(1-\frac{1}{b\delta}\right)^{b-l}\cdot f_{avg}(\pi_b)=b\delta\left[1-\left(1-\frac{1}{b\delta}\right)^b\right]\cdot f_{avg}(\pi_b)
\end{equation*}
Similarly, the right hand side of Inequality (17) can be written as follows:
\begin{equation}
\sum_{l=1}^{b}\left[f_{avg}(\pi_l^g)+b\delta\cdot\Delta_{l+1}\right]\left(1-\frac{1}{b\delta}\right)^{b-l}
\end{equation}
Then, (19) can be reduced to
\begin{flalign}
&=\sum_{l=1}^{b}\left[\sum_{i=1}^{l}\Delta_l+b\delta\cdot\Delta_{l+1}\right]\left(1-\frac{1}{b\delta}\right)^{b-l}\nonumber\\
&=\sum_{l=1}^{b}\left(1-\frac{1}{b\delta}\right)^{b-l}\cdot\sum_{i=1}^{l}\Delta_l+b\delta\sum_{l=1}^{b}\left(1-\frac{1}{b\delta}\right)^{b-l}\cdot\Delta_{l+1}\nonumber\\
&=\sum_{l=1}^{b}\left(1-\frac{1}{b\delta}\right)^{b-l}\cdot\sum_{i=1}^{l}\Delta_l+b\delta\sum_{l=2}^{b+1}\left(1-\frac{1}{b\delta}\right)^{b-l+1}\cdot\Delta_l\nonumber
\end{flalign}
Rearrange and separate the term $\Delta_l$, then combine its coefficients, we have
\begin{flalign}
&=\sum_{l=1}^{b}\left[\sum_{j=l}^b\left(1-\frac{1}{b\delta}\right)^{b-j}+b\delta\left(1-\frac{1}{b\delta}\right)^{b-l+1}\right]\cdot\Delta_l\nonumber\\
&+b\delta\left[\left(1-\frac{1}{b\delta}\right)^{b-(b+1)+1}\cdot\Delta_{b+1}-\left(1-\frac{1}{b\delta}\right)^{b-1+1}\cdot\Delta_{1}\right]\nonumber
\end{flalign}
Here, if we assume
\begin{equation}
\Delta_{b+1}\leq\left(1-\frac{1}{b\delta}\right)^b\cdot\Delta_1
\end{equation}
Then, we have
\begin{flalign}
&\leq\sum_{l=1}^{b}\left[\sum_{j=l}^b\left(1-\frac{1}{b\delta}\right)^{b-j}+b\delta\left(1-\frac{1}{b\delta}\right)^{b-l+1}\right]\cdot\Delta_l\nonumber\\
&=\sum_{l=1}^{b}b\delta\left[\left(1-\left(1-\frac{1}{b\delta}\right)\right)^{b-l+1}+\left(1-\frac{1}{b\delta}\right)^{b-l+1}\right]\cdot\Delta_l\nonumber\\
&=\sum_{l=1}^{b}b\delta\cdot\Delta_l=b\delta\cdot f_{avg}(\pi_b^g)\nonumber
\end{flalign}
Therefore, combining with above, we have
\begin{equation}
f_{avg}(\pi_b^g)\geq\left[1-\left(1-\frac{1}{b\delta}\right)^b\right]\cdot f_{avg}(\pi_b)
\end{equation}
Then, we can obtain the conclusion, Inequality (16), when $b\rightarrow\infty$ from Inequality (21).
\end{proof}

\subsection{Further Estimation}
From the Lemma 5 in the last subsection, it gives us an upper bound $\delta$ of $\Gamma(i|\psi',\psi)$ for any $i$, $\psi$ and $\psi'$. However, this estimation is too rough to get a valid approximation ratio. A natural question is whether we can obtain a tight upper bound. For $\max_{\psi}\Delta(i|\psi,s(i)=1)$, there exists a more precise solution which depends on the structure of network dataset. Given a user $i\in V$, all the potential following participants can be found by breath-first searching, here we assume the edge probability in network is equal to $1$. Then, we can determine the hop of these followers from this initiator $i$. The maximum revenue from user $i$ and her following participants can be known. Finally, we choose the largest one among all users as the upper bound of $\max_{\psi}\Delta(i|\psi,s(i)=1)$. Given targeted network $G=(V,E)$ and revenue vector $\vec{R}$, the calculation steps are as follows:
\begin{enumerate}
	\item For each user $i\in V$, we find all potential followers by breath-first searching, defined as $N_i=\{S_0,S_1,...,S_k\}$, where $S_0=\{i\}$ and $S_j$, $j\in\{1,2,...,k\}$, is the set of potential j-hop participants to user $i$. It means that the length of the shortest path from a user in $S_j$ to user $i$ is equal to $j$.
	\item For each user $i\in V$, we compute the maximum possible revenue from $i$ and her followers, which can be expressed as $P_i$, where
	\begin{equation}
	P_i=\sum_{j=0}^{k}R_i\cdot|S_i|
	\end{equation}
	\item Finally, we select the largest $P_{max}=\max_{i\in V}\{P_i\}$ as the upper bound of $\max_{\psi}\Delta(i|\psi,s(i)=1)$.
\end{enumerate}

Thus, from above process, we can get a estimated vector $\vec{P}=(P_1,P_2,...,P_n)$, which can be computed directly according to the graph dataset. Then, we have
\begin{lem}
	The adaptive total primal curvature $\Gamma(i|\psi',\psi)$ for any $i$, $\psi$ and $\psi'$ is upper bounded by $\delta$, $\Gamma(i|\psi',\psi)\leq\delta$ in the general cases for RMSB problem, where
	\begin{equation}
	\delta=\max\{P_i:P_i\in\vec{P}\}\cdot(R_0-R_1)^{-1}
	\end{equation}
\end{lem}
\noindent
The value of $\delta$ is data-dependent. When $k$ is smaller and the degree distribution is more uniform, the upper bound $\delta$ is more tight. It is in line with our intuition.

\section{Experiment}
In this section, we conduct several experiments to validate the correctness and efficiency of our proposed algorithms on real datasets. It evaluates Algorithm \ref{a1} and \ref{a2} with some common used baseline algorithms.

\subsection{Dataset Description and Statistics}
Our experiments are relied on the datasets from networkrepository.com \cite{nr}, an online network repository. There are three datasets used in our experiments: (1) Dataset-1: A co-authorship network, where each edge is a co-authorship among scientists to publish papers in the area of network theory. (2) Dataset-2: A Wiki network, which is a who-votes-on-whom network collected from Wikipedia. (3) A Collaboration network extracted from Arxiv General Relativity. The statistics information of the three datasets is represented in table \ref{table_1}.

\begin{table}[h]
	\renewcommand{\arraystretch}{1.3}
	\caption{The statistics of three datasets}
	\label{table_1}
	\centering
	\begin{tabular}{|c|c|c|c|c|}
		\hline
		\bfseries Dataset & \bfseries n & \bfseries m & \bfseries Type & \bfseries Average degree\\
		\hline
		Dataset-1 & 0.4K & 1.01K & undirected & 4\\
		\hline
		Dataset-2 & 1.0K & 3.15K & undirected & 6\\
		\hline
		Dataset-3 & 5.2K & 14.5K & undirected & 5\\
		\hline
	\end{tabular}
\end{table}

\subsection{Experimental Settings}
As mentioned earlier, our proposed algorithms are based on the following parameters: Hop number $k$ (at most k-
hop participants follow an initiator), Acceptance vector $\vec{\theta}$, Revenue vector $\vec{R}$, budget $b(\vec{b})$ and edge probability. There are two experiments we perform to test Algorithm \ref{a1} and \ref{a2} used to solve RMSB and RMCB problem, called submodualr performance and non-submodular performance. Shown as Lemma 3, the objective function of RMSB and RMCB is adaptive submodular when $k\leq 1$ or $p_e=1$ for all $e\in E$. Thus, for submodular performance, we set (1) edge probability $p_e=0.5$ for all $e\in E$ and $k=1$; (2) edge probability $p_e=1$ for all $e\in E$ and $k=2$. And in Algorithm \ref{a1} and \ref{a2}, we estimate the Conditional Expected Marginal Benefit $\Delta(u|\phi)$ for each $u\in V$ by the simplifed method of Section 4.2 in \cite{guo2019khop}. For non-submodular performance, we set (1) edge probability $p_e=0.5$ for all $e\in E$ and $k=2$; (2) edge probability  $p_e=0.5$ for all $e\in E$ and $k=3$. We conduct these two experiments on the three datasets shown as above, and analyze their experimental results.

\begin{figure}[!t]
	\centering
	\includegraphics[width=3.5in]{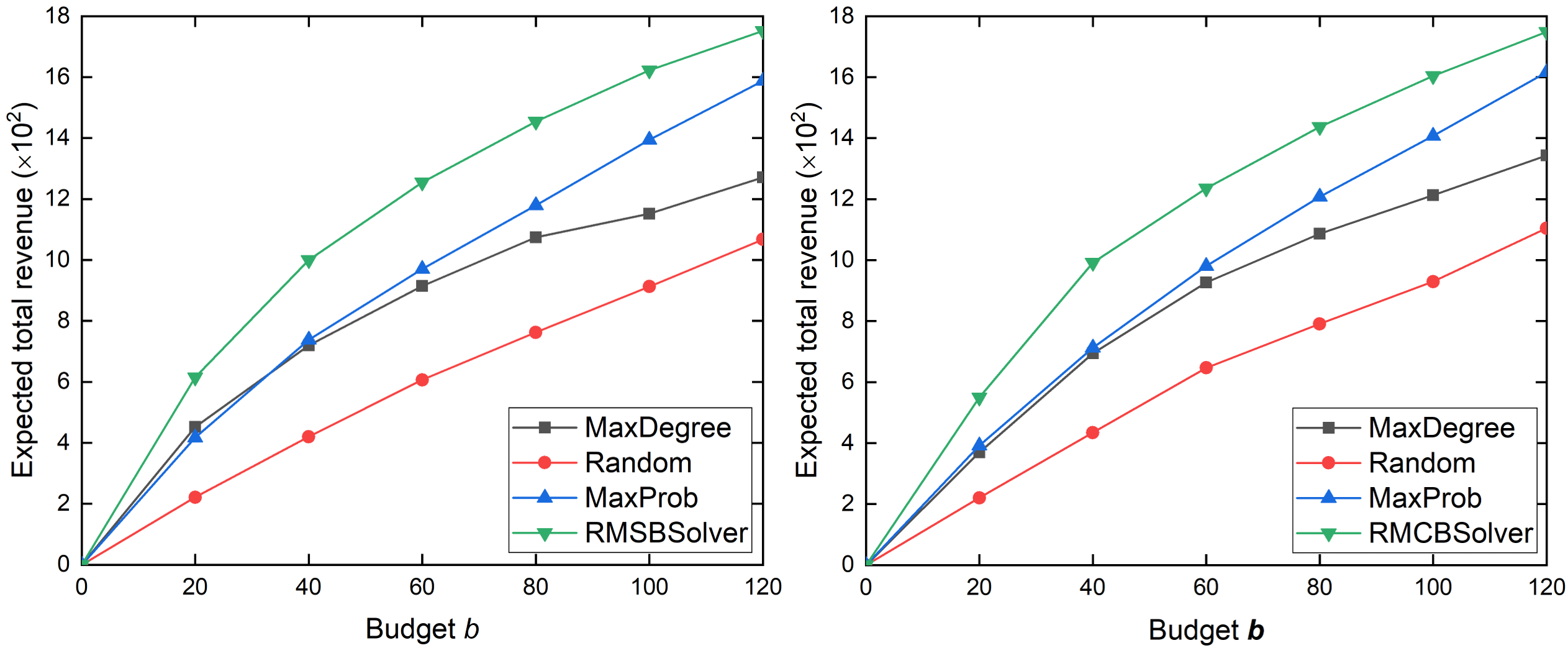}
	\\(a) Dataset-1
	\\${}$
	\includegraphics[width=3.5in]{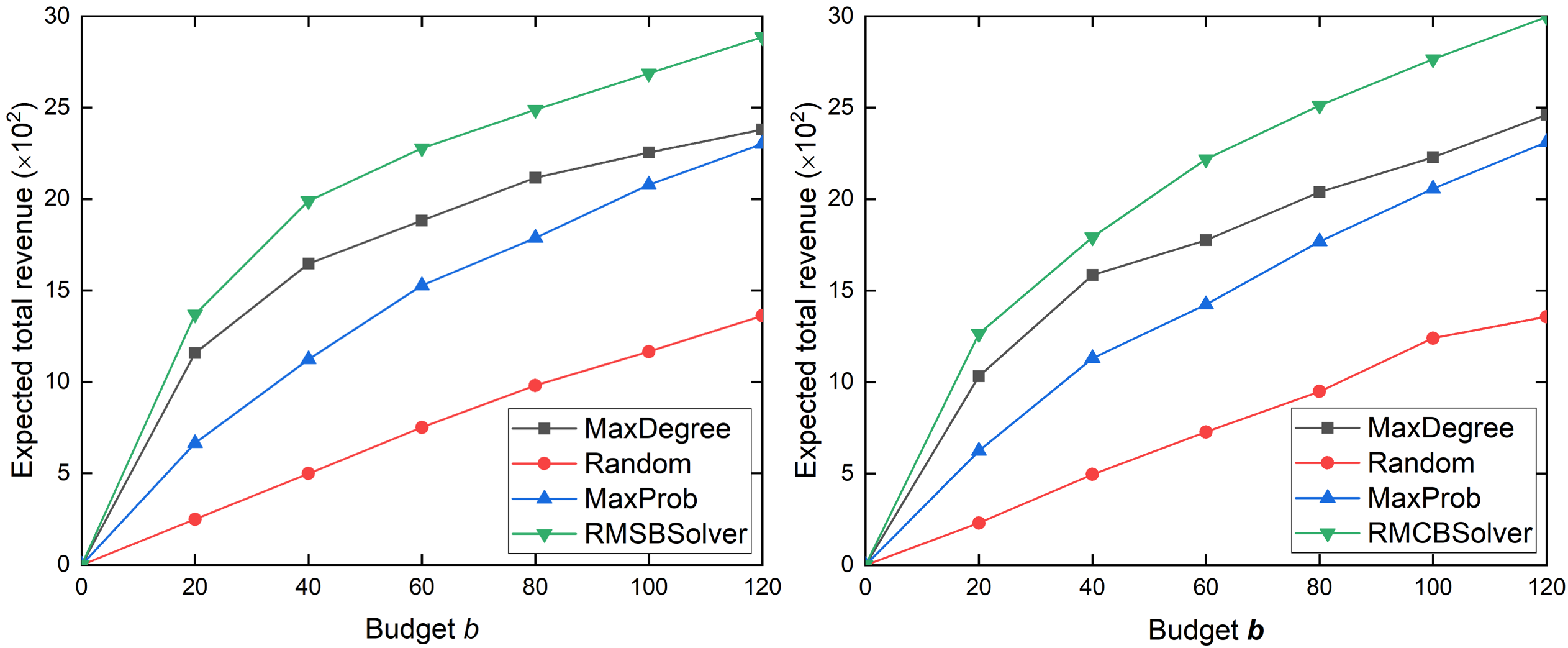}
	\\(b) Dataset-2
	\\${}$
	\includegraphics[width=3.5in]{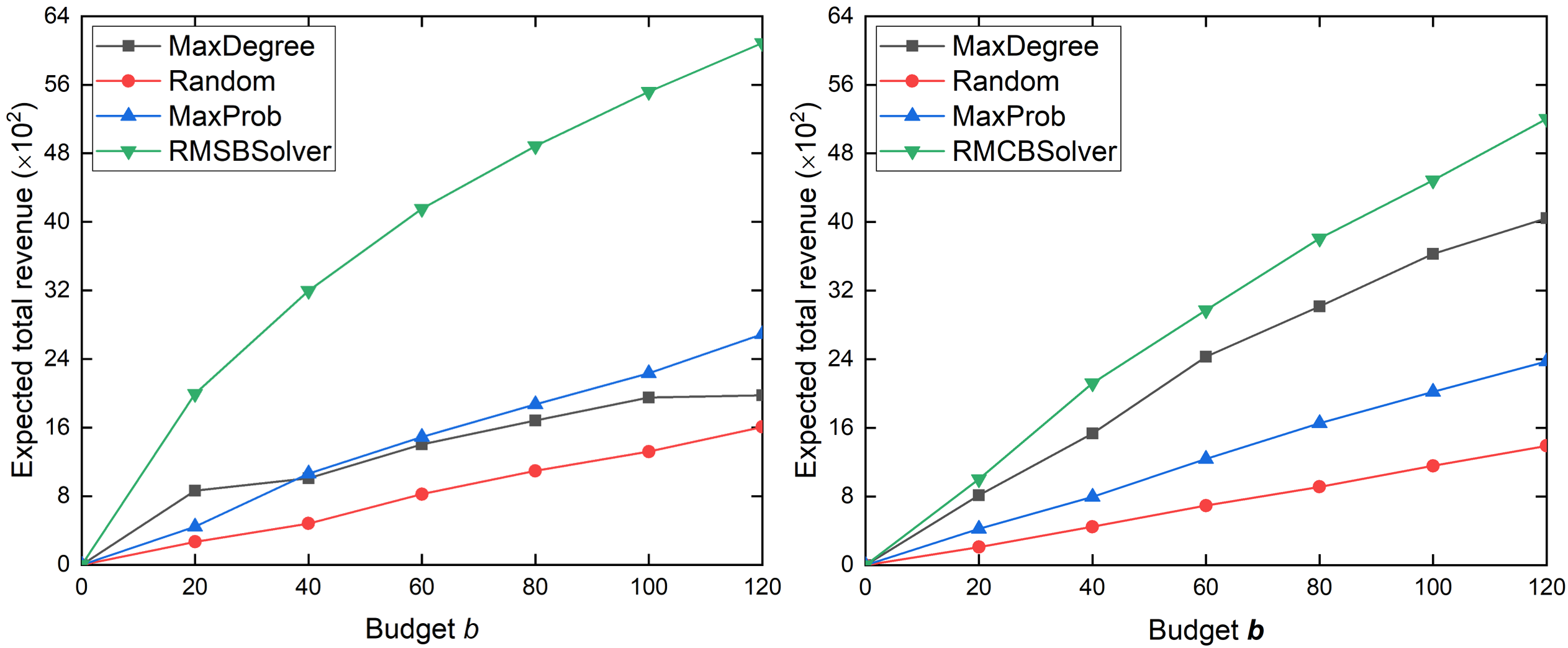}
	\\(c) Dataset-3
	\caption{Submodular performance: the performance comparison of algorithms with different budgets on these three datasets. Parameter settings are $p_e=0.5$ for each $e\in E$ and $k=1$. Left column is to solve RMSB problem, and right column is to solve RMCB.}
	\label{fig2}
\end{figure}

\begin{figure}[!t]
	\centering
	\includegraphics[width=3.5in]{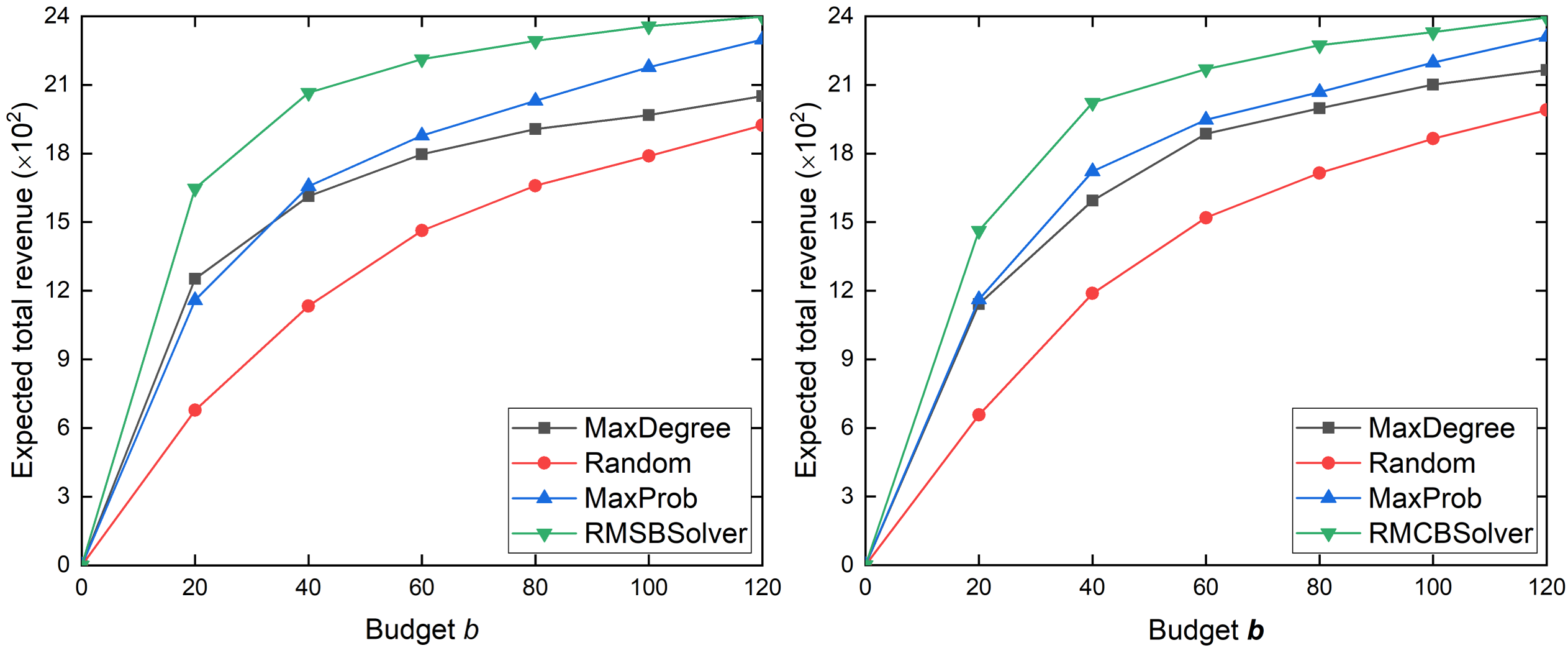}
	\\(a) Dataset-1
	\\${}$
	\includegraphics[width=3.5in]{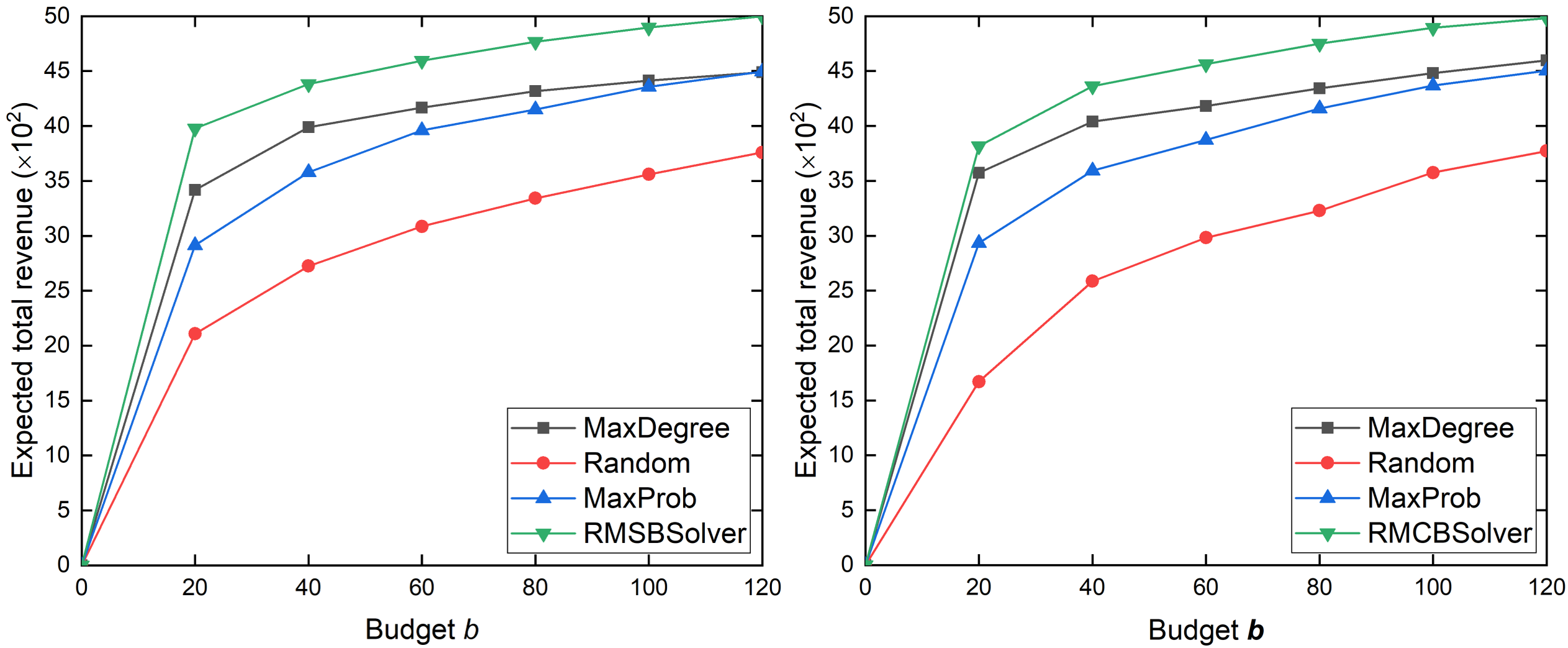}
	\\(b) Dataset-2
	\\${}$
	\includegraphics[width=3.5in]{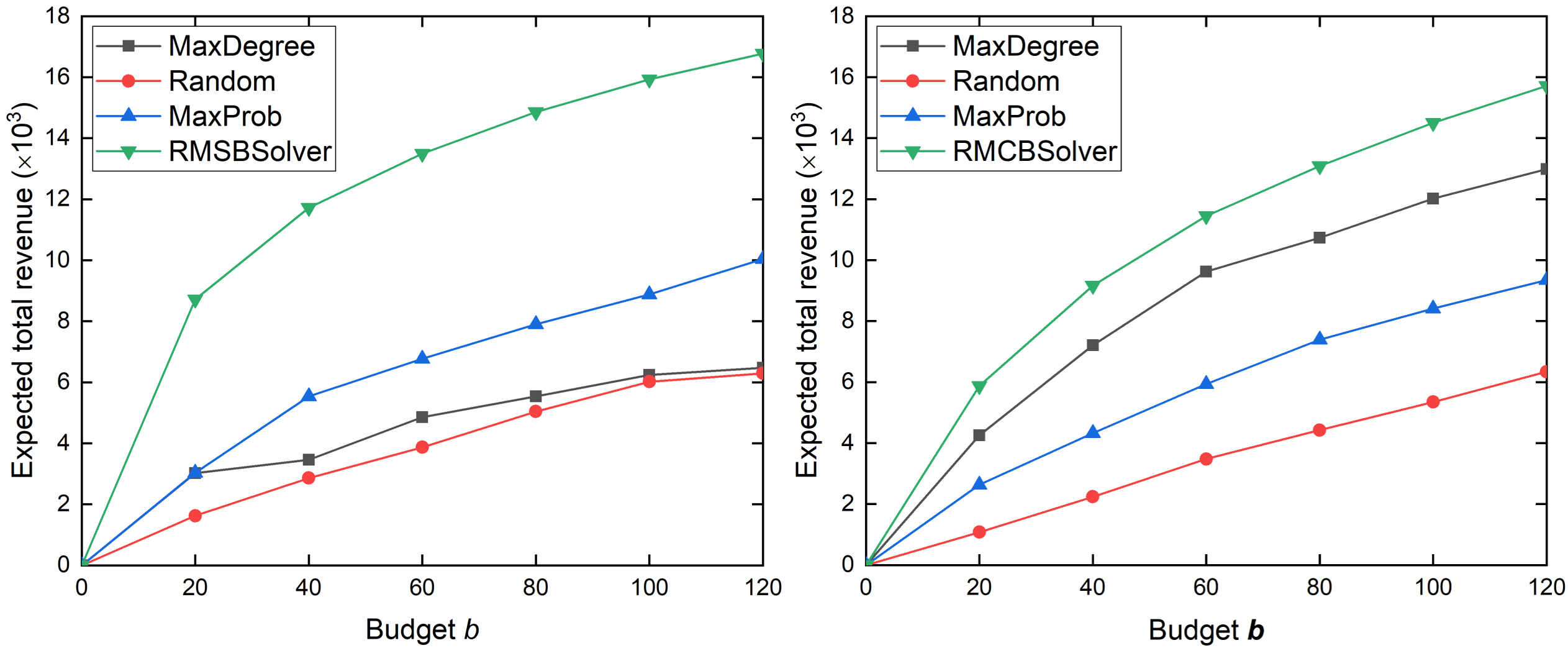}
	\\(c) Dataset-3
	\caption{Submodular performance: the performance comparison of algorithms with different budgets on these three datasets. Parameter settings are $p_e=1$ for each $e\in E$ and $k=2$. Left column is to solve RMSB problem, and right column is to solve RMCB.}
	\label{fig3}
\end{figure}

\begin{figure}[!t]
	\centering
	\includegraphics[width=3.5in]{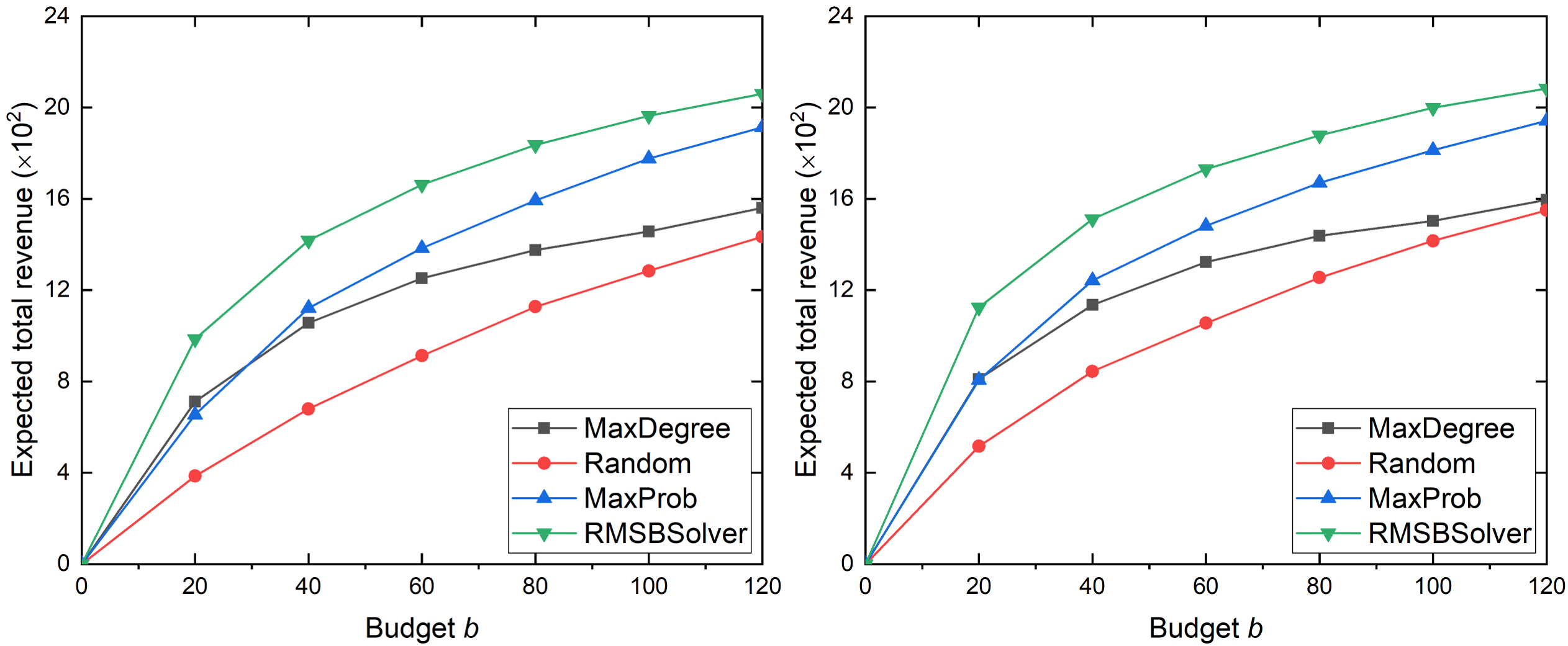}
	\\(a) Dataset-1
	\\${}$
	\includegraphics[width=3.5in]{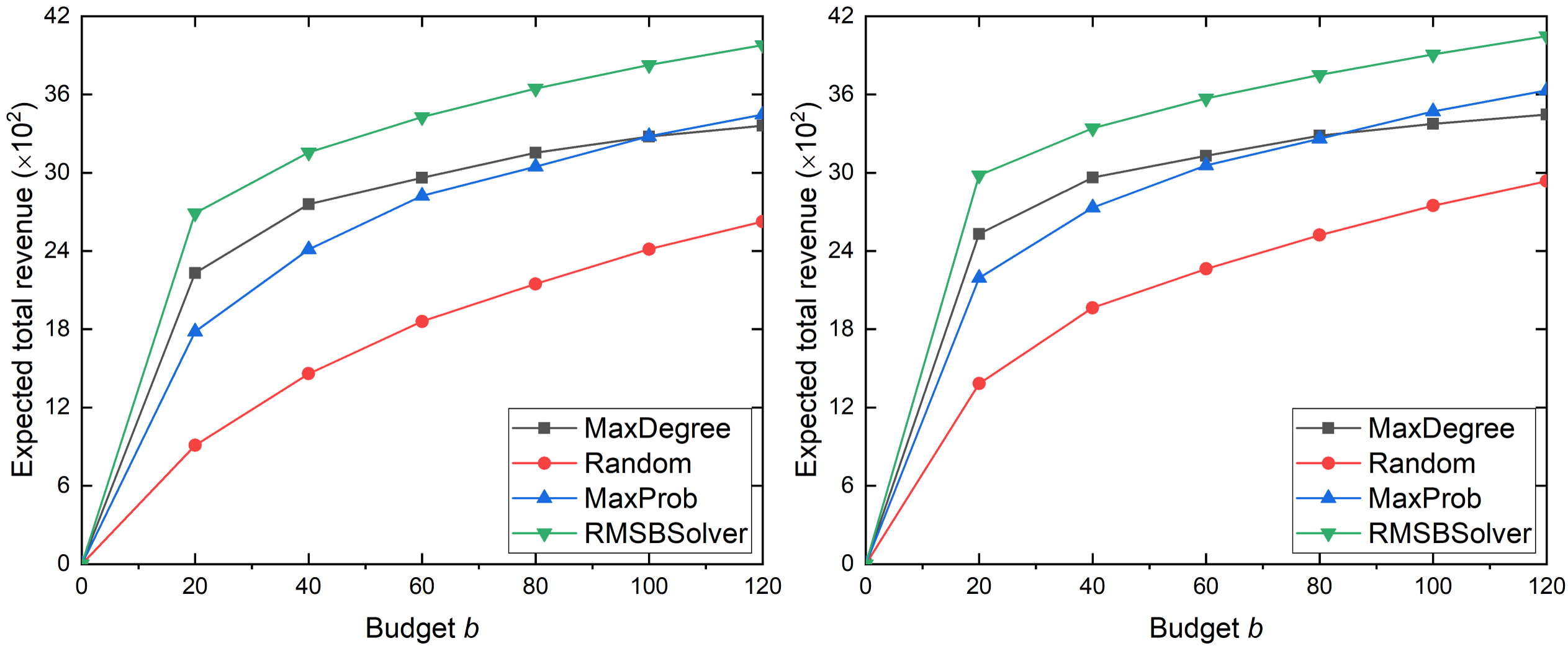}
	\\(b) Dataset-2
	\\${}$
	\includegraphics[width=3.5in]{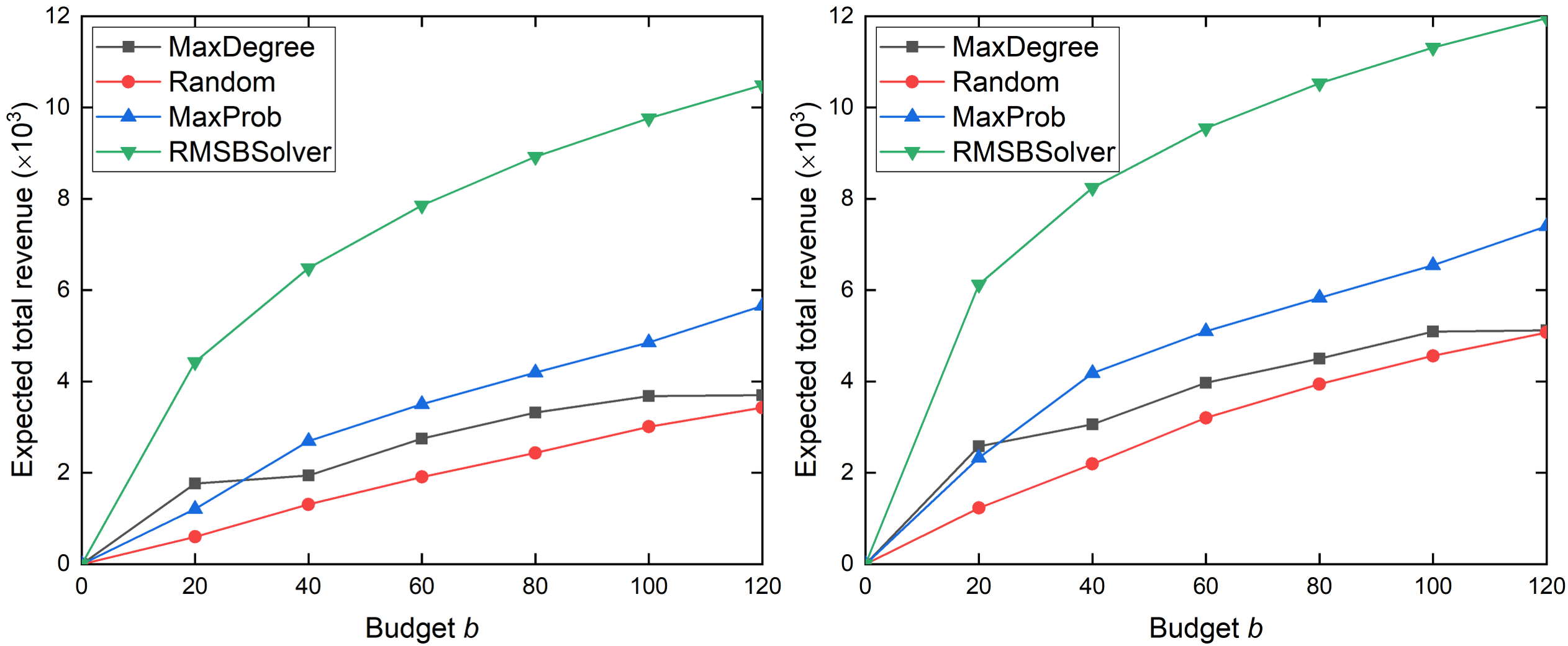}
	\\(c) Dataset-3
	\caption{Non-submodular performance: the performance comparison of algorithms to solve RMSB problem with different budgets on these three datasets. Left column: parameter settings are $p_e=0.5$ for each $e\in E$ and $k=2$. Right column: parameter settings are $p_e=0.5$ for each $e\in E$ and $k=3$.}
	\label{fig4}
\end{figure}

For submodular performance and non-submodular performance, we observe the outcomes of our Algorithm \ref{a1}, Algorithm \ref{a2} and some common used heuristic algorithms, then we compare the results of their performance. It aims to evaluate the effectiveness of the adaptive greedy policy on adaptive (non-)submodular cases under the size and community budget. The Revenue vector $\vec{R}$ can be set as: (1) $(8,6)$ when $k=1$; (2) $(8,6,4)$ when $k=2$; (3) $(8,6,4,2)$ when $k=3$. It means that the revenue of 0-hop participant is 8 units, 1-hop is 6 units, 2-hop is 4 units and 3-hop is 2 units. The baseline algorithms are shown as follows:
\begin{enumerate}
	\item MaxDegree: Invite the user with maximum degree at each step within budget $b(\vec{b})$.
	\item Random: Invite a user randomly from $V$ at each step within budget $b(\vec{b})$.
	\item MaxProb: Invite the user with maximum acceptance probability at each step within budget $b(\vec{b})$.
\end{enumerate}
We use python programming to test each algorithm. The simulation is run on a Windows machine with a 3.40GHz, 4 core Intel CPU and 16GB RAM.

\subsection{Experimental Results}
In our experiments, the whole datasets are considered as the targeted networks. These adaptive algorithms, including our proposed algorithms and other baseline algorithms, are run 50 times on each network, and we take the average of them as the final results. To size budget, it is easy to understand that the number of invitations is limited as $b$. However, for RMCB problem, we need to know how to define this community budget $\vec{b}$. Here, we adopt such a strategy: Suppose the total number of invitations in budget $\vec{b}$ is predefined, the budget for each community can be represented as follows: Given $G=(V,E)$ and budget vector $\vec{b}=(b_1,b_2,...,b_r)$, we have
\begin{equation}
b_r=\left\lfloor\frac{|C_r|\cdot(\text{Total invitations of }\vec{b})}{|V|}\right\rfloor
\end{equation}
It is possible that $\sum_{i=1}^{r}b_i$ is not equal to the total invitations of $\vec{b}$. We need to make some adjustments such that $\sum_{i=1}^{r}b_i$ is equal to the total invitation of $\vec{b}$: add one or mimus one to each community's budget from large community to small community until satisfying equality. Let us look at a specific example:
\begin{exmp}
Considering the right columns of Fig. \ref{fig2} and Fig. \ref{fig3}, "Budget $\vec{b}=20$" means the total invitations of $\vec{b}$ is $20$. Suppose that the community structure is $\mathcal{C}(G)=\{C_1,C_2,C_3,C_3\}$, where $|C_1|=50$, $|C_2|=40$, $|C_3|=30$, and $|C_4|=20$, we have $b_1=\lfloor50\times20/140\rfloor=7$, $b_2=5$, $b_3=4$, and $b_4=2$. Now, $\sum_{i=1}^{4}b_i=18<20$, thus, we let $b_1=7+1=8$, $b_2=5+1=7$ and others maintain so that satisfying $\sum_{i=1}^{4}b_i=20$. This is because $C_1$ and $C_2$ are the largest communities.
\end{exmp}

For submodular performance, the experimental results achieved by different algorithms on these three datasets are shown as Fig. \ref{fig2} and Fig. \ref{fig3}. The left columns of Fig. \ref{fig2} and Fig. \ref{fig3} are under the size budget, and we can see that the expected total revenues returned by RMSBSolver is larger than other policies. The right columns of Fig. \ref{fig2} and Fig. \ref{fig3} are under the community budget, and the performance of RMCBSolver is better than other policies as well. Especially on Dataset-3, the gap between RMSBSolver and other baseline algorithms is very significant, which may be related to the size of the dataset and the network structure. The performance of baseline algorithm, such as MaxDegree, is on and off, good for Dataset-2 and bad for Dataset-3. Thus, we cannot predict whether it is good or bad in advance. An interesting discovery is that, intuitively, the total revenue of RMCB should be smaller than that of RMSB, because the constraint of RMCB is more strict and the approximation performance of RMSBSolver is better than that of RMCBSolver. But, from the results, their expected total revenues are very close regardless of adaptive greedy strategy or other heuristic policies, even sometimes the revenues obtained under the community budget are better. For example, on Dataset-3, the performance of MaxDegree algorithm under the community budget is apparently better than that under the size budget. Therefore, considering the distribution of influence and the average revenue comprehensively, community budget is a more sensible choice to us.

For non-submodular performance, the experimental results achieved by different algorithms on these three datasets are shown as Fig. \ref{fig4}. We can see that the expected total revenues returned by RMSBSolver is larger than other policies. Same as before, on Dataset-3, the gap between RMSBSolver and other baseline algorithms is very significant. Even though the objective function of RMSB is not adaptive submodular, from the actual performance, this effect does not seem obvious. In other words, the objective function of RMSB is close to adaptive submodular. According to Lemma 7, the bound $\delta$ of these three datasets are shown as follows:

\begin{table}[h]
	\renewcommand{\arraystretch}{1.3}
	\caption{The bound $\delta$ of three datasets}
	\label{table_2}
	\centering
	\begin{tabular}{|c|c|c|c|}
		\hline
		\bfseries & Dataset-1 & Dataset-2 & Dataset-3\\
		\hline
		$k=2$ & 106 & 310 & 247 \\
		\hline
		$k=3$ & 199 & 1060 & 820 \\
		\hline
	\end{tabular}
\end{table}

\noindent
Unfortunately, this results is frustrating. Even if the bound given by Lemma 7 is much smaller than that given by Lemma 5, it is still too large to get a satisfactory approximation ratio. Thus, we need to look for a more compact estimation method of bound $\delta$ further so as to get a practical and meaningful approximation ratio.

\section{Conclusion}
In this paper, we propose RMSB problem based on Collaborate Game model. In order to consider the distribution of influence and total revenue, we extend RMSB to RMCB problem by use of community budget. We proposed RMSBSolver and RMCBSolver to address both of them according to the adaptive greedy policy. The objective function of RMSB and RMCB is adaptive monotone and not adaptive submodular, but adaptive submodular in some special cases. Then, we reduce the community budget of RMCB to partition matroid, which can be solved within $(1/2)$-approximation under the special submodular cases. For RMSB problem under the general non-submodular cases, we give a data-dependent $(1-e^{-\frac{1}{\delta}})$-approximation through bounding the adaptive total primal curvature by $\delta$. The good performance of our algorithms is verified by our experiments on three real network datasets. The performance under the community budget is not worse than that under the size budget, thus, community budget is a better choice. However, the bound $\delta$ is not satisfactory, in future, we need to improve it further such that getting a smaller one.


%

\ifCLASSOPTIONcompsoc
  \section*{Acknowledgments}
\else
  \section*{Acknowledgment}
\fi

This work is partly supported by National Science Foundation under grant 1747818.

\ifCLASSOPTIONcaptionsoff
  \newpage
\fi



%

\bibliographystyle{IEEEtran}
\bibliography{references}

%




\end{document}